\documentclass[a4,10pt,twocolumn]{scrartcl}

\pdfoutput=1

\usepackage[top=2.5cm, bottom=2.5cm, left=1.6cm, right=1.6cm]{geometry}
\pdfoutput=1

\input{myStandardPreamble_arXiv.tex}
\usepackage[english]{babel}
\usepackage[utf8x]{inputenc}
\usepackage{enumitem}					

\usepackage{libertine}
\usepackage[T1]{fontenc}
\usepackage[noBBpl]{mathpazo}			

\usepackage[hang]{footmisc}
\setfnsymbol{wiley}

\usepackage[T1]{fontenc}        

\usepackage{graphicx}           
\usepackage{amsmath,amssymb}    
\usepackage{color}
\usepackage{xcolor}   
\usepackage{microtype} 

\usepackage{marginnote}         
\usepackage{cite}               

\usepackage{mathtools}

\makeatletter
\renewcommand{\@biblabel}[1]{[#1]\hfill}
\makeatother

\makeatletter
\let\NAT@parse\undefined
\makeatother

\usepackage[format=plain,			
labelsep=period,		
font=small,			
labelfont=bf,			
skip=5pt			
]{caption}

\usepackage{authblk}

\newcommand{\newmarkedtheorem}[1]{%
	\newenvironment{#1}
	{\pushQED{\oprocend}\csname inner@#1\endcsname}
	{\popQED\csname endinner@#1\endcsname}%
	\newtheorem{inner@#1}%
}

\newtheorem{theorem}{Theorem}
\newtheorem{lemma}{Lemma}
\newtheorem{corollary}{Corollary}
\newtheorem{assumption}{Assumption}
\newtheorem{definition}{Definition}
\newtheorem{remark}{Remark}

\makeatletter
\def\th@plain{%
	\thm@notefont{}
	\normalfont 
}
\def\th@definition{%
	\thm@notefont{}
	\normalfont 
}
\makeatother

\addtokomafont{paragraph}{\itshape}

\setkomafont{section}{\Large}
\setkomafont{subsection}{\large}
\setkomafont{subsubsection}{\itshape}

\usepackage{tikz}				
\usetikzlibrary{arrows}				
\usetikzlibrary{shapes.misc}
\usetikzlibrary{patterns}
\usepackage{pgfplots}
\usetikzlibrary{calc}

\colorlet{istblue}{blue} 
\colorlet{istred}{red!90!black}
\colorlet{istorange}{orange}
\colorlet{istgreen}{green!50!black}

\newcommand{\X}{\mathbb{X}}
\newcommand{\U}{\mathbb{U}}
\newcommand{\Z}{\mathbb{Z}}
\newcommand{\I}{\mathbb{I}}
\newcommand{\R}{\mathbb{R}}

\newcommand{\pushright}[1]{\ifmeasuring@#1\else\omit\hfill$\displaystyle#1$\fi\ignorespaces}

\begin{document}

\title{Economic MPC using a Cyclic Horizon with Application to Networked Control
	Systems} 

\date{}
\author[1]{Stefan Wildhagen} 
\author[2]{Matthias A. M\"uller} 
\author[1]{Frank Allg\"ower}
\affil[1]{Institute for Systems Theory and Automatic Control, University of Stuttgart, 
   Germany. \newline e-mail: \{wildhagen,allgower\}@ ist.uni-stuttgart.de.}
\affil[2]{Institute of Automatic Control, Leibniz University Hannover, Germany. e-mail: mueller@irt.uni-hannover.de}

\maketitle

\begin{abstract}                
\textbf{Abstract.} In this paper, we analyze an economic model predictive control scheme with terminal region and cost, where the system is optimally operated in a certain subset of the state space. The predictive controller operates with a cyclic horizon, which means that starting from an initial length, the horizon is reduced by one at each time step before it is restored to its maximum length again after one cycle. We give performance guarantees for the closed loop, and under a suitable dissipativity condition, establish convergence to the optimal subset. Moreover, we present conditions under which asymptotic stability of the optimal subset can be guaranteed. The results are illustrated in a practical example from the context of Networked Control Systems, which initially motivated the development of the theory presented in this paper.
\end{abstract}


\section{Introduction}

In classical or stabilizing model predictive control (MPC), the objective function is typically designed so as to stabilize a set point or trajectory, which has been pre-chosen under consideration of the system's economic targets. In contrast, economic MPC aims to optimize the economic operating cost of the system directly, by explicitly considering it in its objective function (see, e.g., the survey \cite{Faulwasser18}). This economic cost, however, might be arbitrary such that the optimal regimes of operation for the system can be more complex than to remain in a set point. Optimal modes of operation could in addition comprise periodic orbits or general control invariant subsets of the state space. An important research question is hence to characterize the performance of the closed loop and to determine whether it converges to these general optimal regimes of operation. 

In the context of optimal steady-state operation, a number of different approaches have been considered: terminal equality constraints (\cite{Angeli12}), terminal region and cost (\cite{Amrit11}) and also MPC without terminal conditions (\cite{Gruene13,Gruene14}). The former two schemes achieve stability while for the latter, practical stability is shown. For optimal operation on a periodic orbit, e.g., \cite{Zanon17} and \cite{Mueller16} devise MPC schemes to ensure (practical) stability of the optimal periodic orbit. The case that the optimal mode of operation is a general control invariant subset has also recently received attention in the literature. Stability of the optimal subset is shown in \cite{Martin19} for an MPC with a terminal equality constraint (in the sense that the terminal state lies somewhere in the optimal subset). Conditions under which the optimal subset is stabilized for an MPC with terminal region and cost are provided in \cite{Dong18}. In all of the preceding results, a dissipativity property of the system with respect to the optimal regime of operation is required for stability.

Typically in MPC, a fixed horizon is used, although a variable horizon may be beneficial under certain circumstances. A cyclic horizon, in particular, describes a horizon that shrinks from a maximum length in each sampling time step until it reaches a minimum length, and that is then restored to its maximum value before the cycle is started again. Stabilizing MPC schemes with a cyclic horizon were analyzed in \cite{Koegel13} and \cite{Lazar15}. In both references, an MPC with a terminal region and cost was considered and it was shown that a cyclic horizon can be leveraged to use more flexible terminal regions: Instead of requiring the terminal region to be control invariant, it is sufficient that the state is able to return there after several time steps. However, both references rely on a uniform upper bound on the value function of the MPC optimization problem to establish stability, a condition which might be difficult to verify for general nonlinear systems. Multi-step MPC, which is equivalent to a cyclic horizon in the absence of disturbances, was furthermore considered for schemes without terminal conditions, e.g., in the stabilizing setup in \cite{Gruene09} or in the context of economic MPC with optimal periodic operation in \cite{Mueller16}.

In this note, we consider economic MPC with \textit{cyclic horizons}, where the optimal regime of operation is a \textit{general control invariant subset} of the state space. To ensure recursive feasibility and stability of the closed loop, we use a \textit{terminal region and a terminal cost} in the MPC optimization problem. The combination of economic MPC and a cyclic horizon was motivated by a problem arising in the context of Networked Control Systems (NCS). In the considered setup, as detailed in Section \ref{sect_NCS}, transmission of a control input over a dynamical network can only be guaranteed after a certain number of time steps.

We contribute to the existing theory first by providing a performance and convergence analysis for the considered economic setup with cyclic horizons. Second, we show that a suitable upper bound on the value function at initial time is sufficient to guarantee stability. Furthermore, we develop constructive conditions that guarantee this upper bound, similar to what is already known in the standard case with a constant horizon. Such an analysis is novel also for stabilizing MPC with cyclic horizons, where a uniform upper bound of the value function was formerly required to establish stability (\cite{Koegel13,Lazar15}). Also in this context, the requirement that the invariant set lies in the interior of the terminal region, as typically used in MPC, is attenuated to a more general condition.

The remainder of this paper is organized as follows. In Section \ref{sect_preliminaries}, we introduce some notation and the considered setup. Convergence and performance properties of the economic MPC with cyclic horizon are treated in Section \ref{sect_convergence}, while Section \ref{sect_stability} is devoted to conditions for asymptotic stability. We present the NCS application mentioned above as a special case in Section \ref{sect_NCS} to illustrate our main results.

\section{Preliminaries} \label{sect_preliminaries}

\subsection{Notation}

Let $\I$ and $\mathbb{R}$ denote the set of all integers and real numbers, respectively. We denote $\I_{[a,b]}\coloneqq\I\cap[a,b]$ and $\I_{\ge a}\coloneqq\I\cap[a,\infty)$, $a,b\in\I$, and $\mathbb{R}_{\ge a}\coloneqq[a,\infty)$, $a\in\R$. A function $\alpha:\mathbb{R}_{\ge 0}\rightarrow\mathbb{R}_{\ge 0}$ is said to be of class $\mathcal{K}_\infty$ if it is continuous, zero at zero, strictly increasing and unbounded. We denote by $I$ the identity matrix. For a vector $v\in\mathbb{R}^n$, the set distance to a subset $A\subseteq\mathbb{R}^n$ is defined as $||v||_A \coloneqq \min_{w\in A} ||v-w||$. The Minkowski set addition of two sets $A,B\subset\mathbb{R}^n$ is defined by $A\oplus B \coloneqq \{v\in\mathbb{R}^n|\exists a\in A,b\in B: v=a+b\}$. A ball of radius $a$ around the origin is defined by $\mathcal{B}_a\coloneqq\{v\in\mathbb{R}^n|||v||_2\le a\}$.

\subsection{General Setup}

We consider the nonlinear discrete-time system
\begin{equation}
x(k+1) = f(x(k),u(k)), \label{system_dyn}
\end{equation}
where $x(k) \in \X \subseteq \mathbb{R}^n$ is the system state and $u(k) \in \U \subseteq \mathbb{R}^m$ the controlled input at time $k\in \I_{\ge 0}$. Both the state and input constraint sets $\X$ and $\U$ are assumed to be closed. The state and input are subject to the mixed constraints
\begin{equation*}
(x(k),u(k)) \in \Z \subseteq \X \times \U, \quad k\in \I_{\ge 0},
\end{equation*}
such that $f: \Z \rightarrow \X$, where $f$ is continuous. Note that we do not require $\Z$ to be compact, as is frequently assumed in the economic MPC literature. Associated with the system is an economically motivated, continuous cost $\ell: \Z \rightarrow \mathbb{R}$, which is not assumed to fulfill any definiteness property.

The solution of \eqref{system_dyn} at time $k\in\I_{[0,N]}$, $N\in\I_{\ge 0}$, starting from an initial state $x(0) = x$ and under the input sequence $u(\cdot) = \{u(0),u(1),\ldots,u(N-1)\}\in\U^N$ is denoted by $x_u(x,k)$. Furthermore, $\U^N(x) \coloneqq \{u(\cdot) \in \U^N | (x_u(x,k),u(k))\in\Z, \; \forall k \in \I_{[0,N-1]}\}$ is the set of all admissible input sequences of length $N$ at state $x\in\X$.

As is common in the economic MPC literature, we are interested in the ``long-run'' optimal operation of the system. As in \cite{Dong18}, we define the asymptotic average cost $\ell^*_{av}(x)$ starting from an initial state $x\in\X$.
\begin{definition}
	For a given initial state $x\in \X$, the  best asymptotic average cost is defined as
	\begin{equation*}
	\ell^*_{av}(x) \coloneqq \inf_{u(\cdot) \in \U^K(x)} \liminf_{K\rightarrow\infty}  \frac{\sum_{k=0}^{K-1} \ell(x_u(x,k),u(k))}{K}.
	\end{equation*}
	Moreover, denote by $\ell^*_{av} = \inf_{x\in\X} \ell^*_{av}(x)$ the lowest possible asymptotic average cost.
\end{definition}
Hence, $\ell^*_{av}$ describes the lowest possible average cost along all admissible trajectories of system \eqref{system_dyn}. In the remainder, we assume that the infima are attained. 

The system is controlled by an MPC controller with a cyclically varying horizon $N$ given by
\begin{equation*}
N(k) = \hat{N} - k\text{mod} M,
\end{equation*}
where $M\in\I_{\ge 1}$ denotes the cycle length and $\hat{N}\in\I_{\ge M}$ the maximum horizon length (cf. \cite{Koegel13}). The horizon length is $\hat{N}$ at $k=jM, \; j\in\I_{\ge 0}$ and then shrinks by one in each time step until the minimum horizon length $\hat{N}-M+1$ is reached at $k=jM-1, \; j\in\I_{\ge 1}$. Then, the full horizon $\hat{N}$ is restored at the next time instance and the cycle starts again, as illustrated in Figure \ref{cyclic_horizon}.

\begin{figure}[h]
	\centering
	\begin{tikzpicture}[>=latex]
	\begin{axis}[
			width=\columnwidth,
			height=4cm,
			xmin=0, 
			xmax=8,
			ymin=-1, 
			ymax=5,
			disabledatascaling,
			ylabel=Time $k$,
			xlabel=Prediction horizon $N(k)$,
			xtick={-1,0,1,2,3,4,5,6,7,8},
			xticklabels={,$0$,$1$,$2$,$3$,$4$,$5$,$6$,$7$,$8$},
			ytick={0,1,2,3,4,5},
			yticklabels={$0$,$1$,$2$,$3$,$4$,$5$},
			ylabel style={yshift=-0.3cm},
			xmajorgrids=true,
			axis x line=bottom,
			axis y line=left]
		\addplot[very thick,color=istgreen,mark=*] coordinates {(0,0) (5,0)};
		\addplot[very thick,color=istorange,mark=*] coordinates {(1,1) (5,1)}; 
		\addplot[very thick,color=istblue,mark=*] coordinates {(2,2) (5,2)};
		\addplot[very thick,color=istgreen,mark=*] coordinates {(3,3) (8,3)}; 
		\addplot[very thick,color=istorange,mark=*] coordinates {(4,4) (8,4)}; 
		\addplot[very thick,color=istblue,mark=*] coordinates {(5,5) (8,5)};  

	\end{axis} 
\end{tikzpicture}
	\caption{Cyclic horizon with $\hat{N}=5$ and $M=3$.}
	\label{cyclic_horizon}
\end{figure}
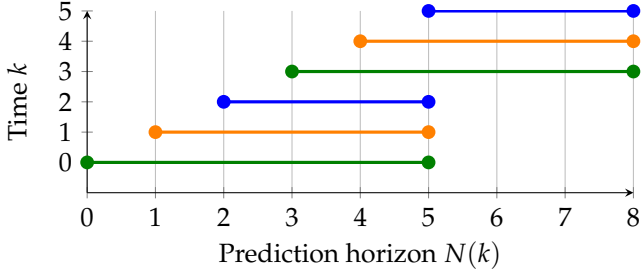

In the economic MPC scheme, the objective function is defined as
\begin{equation}
V(x,u(\cdot),N) \coloneqq \sum_{i=0}^{N-1} \ell(x_u(x,i),u(i)) + V_f(x_u(x,N)) \label{objective_fun}
\end{equation}
with the continuous terminal cost function $V_f: \X_f\rightarrow \mathbb{R}$. Since $f$, $\ell$ and $V_f$ are continuous, $V$ is as well. The set $\X_f\subseteq\X$ denotes a closed terminal region, in which the predicted terminal state is required to be contained, i.e., $x_u(x,N)\in\X_f$. Hence, the MPC optimization problem $\mathbb{P}(x,k)$ solved at state $x$ and time $k$ reads
\begin{equation}
V^*(x,k) \coloneqq \hspace{-7pt} \min_{u(\cdot)\in\U^{N(k)}(x)} \hspace{-3pt} \{V(x,u(\cdot),N(k))| x_u(x,N(k))\in\X_f\}.
\label{value_fun}
\end{equation}
We denote the minimizer of \eqref{value_fun} by
\begin{equation*}
u^*(\cdot;x,k) = \{u^*(0;x,k),\ldots,u^*(N(k)-1;x,k)\}
\end{equation*}
and the feasible set, i.e., the set of all initial states $x$ such that $\mathbb{P}(x,k)$ is feasible with horizon $N(k)$, by $\mathcal{X}_{N(k)}$. Subsequently, the first part of the optimal control sequence is applied to the system according to $u(k)=u^*(0;x,k)\eqqcolon\nu_k(x)$, then the resulting state at $k+1$ is measured and the optimal control problem is solved anew. We further denote $\nu\coloneqq\{\nu_k\}_{k\in\I_{\ge 0}}$.

\begin{remark}
	Since no disturbances act on the system, equivalent to the procedure above is to solve the optimization problem with full horizon $\hat{N}$ every $M$ time steps and to apply the first $M$ pieces of the optimal input trajectory. This follows immediately from Bellman's principle of optimality. In the presence of disturbances, however, such a scheme might reduce performance due to the longer open-loop phase (cf. \cite{Gruene15}).
\end{remark}

\begin{definition}
	System \eqref{system_dyn} is strictly dissipative with respect to a set $X$ and the supply rate $s: \Z\rightarrow\mathbb{R}$ if there exists a storage function $\lambda: \X\rightarrow\mathbb{R}_{\ge 0}$ and a $\mathcal{K}_\infty$-function $\rho$ such that for all $(x,u)\in\Z$
	\begin{equation*}
	\lambda(f(x,u))-\lambda(x) \le s(x,u) - \rho(||x||_X).
	\end{equation*}
\end{definition}

\begin{definition}
	A nonempty set $X\subseteq\X$ is called control invariant if for all $\bar{x}\in X$, there exists a $\bar{u}\in\U$ such that $f(\bar{x},\bar{u})\in X$.
\end{definition}

\begin{assumption}
	System \eqref{system_dyn} is strictly dissipative with respect to the control invariant set $\bar{\X}$ and the supply rate $s(x,u)=\ell(x,u)-\ell^*_{av}$ with a continuous storage function.
	\label{ass_dissi}
\end{assumption}

\begin{remark}
	The notion of a control invariant set $\bar{\X}$ also comprises steady states and periodic orbits.
\end{remark}

\section{Convergence and Performance} \label{sect_convergence}

A well-studied way to ensure recursive feasibility and convergence in economic MPC is to assume that the terminal region is a control invariant set and that the forward difference of the terminal cost in $\X_f$ is bounded by $-\ell+\ell^*_{av}$ under a local control (\cite{Amrit11,Dong18}). In this work, we use slightly relaxed conditions similar to those in \cite{Koegel13,Lazar15}.

In the following, we consider multiple terminal controllers $\kappa_0(x),\ldots,\kappa_{M-1}(x)$, which make the terminal region ``$M$-step invariant''. Given an initial state $x\in\X_f$, the controllers $\kappa_0$ to $\kappa_{M-1}$ are applied successively. After $M$ time steps, $\kappa_0$ is used again and the cycle is restarted. For an initial state $x\in \X_f$, we denote the solution of \eqref{system_dyn} $k$ time steps after the initial time, resulting from an application of this terminal control law, by $x_{\kappa}(x,k)$.

\begin{assumption}
	There exist a terminal region $\X_f$, a cycle length $M$ and terminal control laws $\{\kappa_k\}_{k\in\I_{[0,M-1]}}$ such that for all $x\in\X_f$, it holds that $x_\kappa(x,M)\in\X_f$ and
	\begin{equation*}
	(x_\kappa(x,k),\kappa_k(x_\kappa(x,k))) \in \Z, \quad \forall k\in\I_{[0,M-1]}.
	\end{equation*}
	\label{ass_term_inv}
\end{assumption}

This assumption requires that for an initial state in the terminal region $\X_f$, there must exist a feasible control law that drives the state back into the terminal region after one cycle period $M$. In the meantime, the state is merely not allowed to leave the constraint set. Note that while $\bar{\X}$ is a control invariant set, Assumption \ref{ass_term_inv} only requires the terminal region $\X_f$ to be ``$M$-step control invariant''.
\begin{assumption}
	There exists a terminal cost $V_f$ such that for all $x\in \X_f$, with $\X_f$, $M$ and $\{\kappa_k\}_{k\in\I_{[0,M-1]}}$ from Assumption \ref{ass_term_inv},
	\begin{align}
	V_f(x_\kappa(x,M))-V_f(x) \le -\sum_{k=0}^{M-1}&\ell(x_\kappa(x,k),\kappa_k(x_\kappa(x,k))) \nonumber \\ 
	+ &M \ell^*_{av}. \label{term_decrease}
	\end{align}
	\label{ass_term_decrease}
\end{assumption}

To analyze convergence of the proposed MPC scheme, as is common in economic MPC, we introduce the rotated cost
\begin{equation}
L(x,u) \coloneqq \ell(x,u) + \lambda(x) - \lambda(f(x,u)) - \ell^*_{av} \label{rotated_cost}
\end{equation}
and the rotated terminal cost
\begin{equation}
\bar{V}_f(x) \coloneqq V_f(x) + \lambda(x). \label{rotated_terminal_cost}
\end{equation}
Then, we define the rotated objective function as
\begin{equation*}
\bar{V}(x,u(\cdot),N) \coloneqq \sum_{i=0}^{N-1} L(x_u(x,i),u(i)) + \bar{V}_f(x_u(x,N))
\label{cost_fun_rot}
\end{equation*}
and the rotated optimal control problem $\bar{\mathbb{P}}(x,k)$ as
\begin{equation*}
\bar{V}^*(x,k) \coloneqq \hspace{-7pt} \min_{u(\cdot)\in\U^{N(k)}(x)} \hspace{-3pt} \{\bar{V}(x,u(\cdot),N(k))| x_u(x,N(k))\in\X_f\}.
\label{value_fun_rot}
\end{equation*}
Note that $f$, $\ell$, $\lambda$ and $V_f$ are continuous, and therefore, $L$, $\bar{V}_f$ and $\bar{V}$ are as well.

\begin{corollary}
	If Assumption \ref{ass_dissi} holds, $L(x,u) \ge \rho(||x||_{\bar{\X}})$.
	\label{cor_rot_cost}
\end{corollary}

With the conditions on the terminal region and cost, we can state the following preliminary result.
\begin{lemma}
	If Assumptions \ref{ass_dissi}, \ref{ass_term_inv} and \ref{ass_term_decrease} hold, then
	\begin{itemize}
		\item the rotated optimal control problem $\bar{\mathbb{P}}(x_\nu(x(0),k),k)$ is feasible for all $k\in\I_{\ge 0}$ if $\bar{\mathbb{P}}(x(0),0)$ is feasible,
		\item for the rotated terminal cost $\bar{V}_f$ it holds that for all $x\in\X_f$,
		\begin{equation}
		\bar{V}_f(x_\kappa(x,M))-\bar{V}_f(x) \hspace{-2pt} \le \hspace{-2pt}- \hspace{-5pt} \sum_{k=0}^{M-1} \hspace{-2pt} L(x_\kappa(x,k),\hspace{-1pt}\kappa_k(x_\kappa(x,k))).
		\label{term_decrease_rot}
		\end{equation}
	\end{itemize}
	\label{lemma_term_decrease_rot}
\end{lemma}
\begin{proof}
	Recursive feasibility can be proven exactly as in \cite[Proposition 4]{Koegel13} since Assumption \ref{ass_term_inv} is equivalent to \cite[Assumption 3]{Koegel13}.
	
	For the second part, adding $\lambda(x_\kappa(x,M))-\lambda(x)$ to both sides of \eqref{term_decrease} gives
	\begin{align*}
	&\underbrace{V_f(x_\kappa(x,M)) + \lambda(x_\kappa(x,M))}_{\stackrel{\eqref{rotated_terminal_cost}}{=}\bar{V}_f(x_\kappa(x,M)))} \underbrace{- V_f(x) - \lambda(x)}_{\stackrel{\eqref{rotated_terminal_cost}}{=}-\bar{V}_f(x)} \\
	&\le \underbrace{\lambda(x_\kappa(x,M)) - \lambda(x)}_{\sum_{k=0}^{M-1} \lambda(x_\kappa(x,k+1)) - \lambda(x_\kappa(x,k))} \\ &\hspace{30pt}-\sum_{k=0}^{M-1}\Big(\ell(x_\kappa(x,k),\kappa_k(x_\kappa(x,k))) + \ell^*_{av}\Big) \\
	&\stackrel{\eqref{rotated_cost}}{=} -\sum_{k=0}^{M-1}L(x_\kappa(x,k),\kappa_k(x_\kappa(x,k))).
	\end{align*}
\end{proof}

\begin{assumption}
	The minimum of the rotated terminal cost is $0$. It is attained exactly on $\bar{\X}$, i.e., $\bar{\X}=\arg\min_{x}\bar{V}_f(x)$.
	\label{ass_lb_term_cost_rot}
\end{assumption}
\begin{remark}
	Assumption \ref{ass_lb_term_cost_rot} can be fulfilled for instance if both $V_f$ and $\lambda$ take their minimal values on $\bar{\X}$. It also implies that $\bar{\X}\subseteq\X_f$, since $V_f$ is defined on $\X_f$. Requiring the minimum to be equal to $0$ is without loss of generality.
\end{remark}
With Lemma \ref{lemma_term_decrease_rot} and Assumption \ref{ass_lb_term_cost_rot}, we state the first main result.
\begin{theorem}
	Suppose $x(0)\in\mathcal{X}_{\hat{N}}$. If Assumptions \ref{ass_dissi}-\ref{ass_lb_term_cost_rot} hold, then the optimization problem $\mathbb{P}(x_\nu(x(0),k),k)$ is feasible for all $k\in\I_{\ge 0}$ and $x_\nu(x(0),k)$ converges to $\bar{\X}$ as $k\rightarrow\infty$.
	\label{thm_convergence}
\end{theorem}
\begin{proof}
	In a first step, we assert that the solution sets of $\mathbb{P}(x,k)$ and that of $\bar{\mathbb{P}}(x,k)$ are identical, which means that they yield the same optimizer $u^*(\cdot;x,k)$. To this end, we notice that both optimization problems are subject to the same constraints and establish using similar techniques as in \cite{Amrit11} that for the objective functions it holds that
	\begin{equation}
	\bar{V}(x,u(\cdot),N) = V(x,u(\cdot),N) - N\ell^*_{av} + \lambda(x).
	\label{cost_fcns_diff}
	\end{equation}
	Since the last two terms in \eqref{cost_fcns_diff} depend entirely on constant parameters of the optimization problem, $\bar{V}(x,u(\cdot),N)$ and $V(x,u(\cdot),N)$ differ by a constant and the claim is proven.
	
	The problem $\mathbb{P}(x(0),0)$ is feasible by assumption since $N(0)=\hat{N}$. With the identical solution sets, the rotated problem $\bar{\mathbb{P}}(x(0),0)$ is feasible as well. By Lemma \ref{lemma_term_decrease_rot}, $\bar{\mathbb{P}}(x_\nu(x(0),k),k)$ is then feasible for all $k\in\I_{\ge 0}$, and the same follows immediately for the original problem $\mathbb{P}(x_\nu(x(0),k),k)$.
	
	Due to the identical solution sets, we are also able to use the rotated problem in the analysis of convergence. Notice that due to optimality of the value function,
	\begin{equation}
	\bar{V}^*(x,k) \le \bar{V}(x,\tilde{u}(\cdot;x,k),N(k))
	\label{value_ub_by_cost}
	\end{equation}
	for some feasible control input $\tilde{u}(\cdot;x,k)$. If $k \neq jM-1$, $j\in \I_{\ge 1}$, due to $N(k) = N(k+1)+1$ and  Bellman's principle of optimality, the optimal input at $k+1$ is
	\begin{align*}
	u^*(\cdot;x_{\nu_k}(x,1),k+1) = \{u^*&(1;x,k),\ldots, \\
	&u^*(N(k+1);x,k)\}.
	\end{align*}
	Thus, from \eqref{objective_fun} we obtain for $k \neq jM-1$
	\begin{equation}
	\bar{V}^*(x_{\nu_k}(x,1),k+1) = \bar{V}^*(x,k) - L(x,u^*(0;x,k)).
	\label{cost_decrease_rot}
	\end{equation}
	For $k = jM-1$, $j\in \I_{\ge 1}$, we have $N(k+1)-N(k)=M$. Denoting $x^{*}\coloneqq x_{u^*(\cdot;x,k)}(x,N(k))$, we choose the feasible input
	\begin{align*}
	\tilde{u}&(\cdot;x_{\nu_k}(x,1),k+1) = \{u^*(1;x,k),\ldots, \\
	&u^*(N(k)-1;x,k), \nonumber \kappa_0(x^{*}),\ldots,\kappa_M(x_{\kappa}(x^{*},M-1))\}.
	\end{align*}
	With this, it holds that
	\begin{align*}
	&\bar{V}(x_{\nu_k}(x,1),\tilde{u}(\cdot;x_{\nu_k}(x,1),k+1),N(k+1)) \\
	&\hspace{-1pt} =  \bar{V}^*(x,k) - L(x,u^*(0;x,k)) + \bar{V}_f(x_\kappa(x^*,M)) \\
	& - \hspace{-1pt} \bar{V}_f(x^*) + \sum_{i=0}^{M-1} L(x_\kappa(x^{*},i),\kappa_i(x_\kappa(x^{*},i)))
	\end{align*}
	and then with \eqref{term_decrease_rot} and \eqref{value_ub_by_cost}
	\begin{equation}
	\bar{V}^*(x_{\nu_k}(x,1),k+1) \hspace{-1pt} \le \hspace{-1pt} \bar{V}^*(x,k) - L(x,u^*(0;x,k)).
	\label{cost_decrease_rot_2}
	\end{equation}
	In view of \eqref{cost_decrease_rot}, we conclude that \eqref{cost_decrease_rot_2} holds for all $k\in\I_{\ge 0}$.
	
	Due to \eqref{cost_decrease_rot_2}, Assumption \ref{ass_dissi} and Corollary \ref{cor_rot_cost}, $\bar{V}^*$ decreases unless $x\in\bar{\X}$. Since $\bar{V}^*$ is lower bounded due to Assumption \ref{ass_lb_term_cost_rot}, it must converge to a constant value and hence, $L\rightarrow 0$ as $k\rightarrow\infty$. From the lower bound on $L$ by Corollary \ref{cor_rot_cost}, it must also hold that $||x_\nu(x(0),k)||_{\bar{\X}}\rightarrow 0$ as $k\rightarrow\infty$.
\end{proof}

In the second main result, we dwell on the performance of system \eqref{system_dyn} in closed loop with the MPC.
\begin{theorem}
	Suppose $x(0)\in\mathcal{X}_{\hat{N}}$. If $\ell$ is bounded on $\Z$ and Assumptions \ref{ass_term_inv} and \ref{ass_term_decrease} hold, then the asymptotic average cost of system \eqref{system_dyn} controlled by the MPC with cyclic horizon, is less than or equal to the lowest possible asymptotic average cost, i.e.,
	\begin{equation*}
	\limsup_{K\rightarrow\infty} \frac{\sum_{k=0}^{K-1} \ell(x_{\nu}(x(0),k),\nu_k(x_{\nu}(x(0),k)))}{K} \le \ell^*_{av}.
	\end{equation*}
	\label{thm_performance}
\end{theorem}
\begin{proof}
	Using the same methodology as in the proof of Theorem \ref{thm_convergence}, we obtain with Assumption \ref{ass_term_decrease}
	\begin{equation*}
	V^*(x_{\nu_k}(x,1),k+1) \hspace{-1pt} \le \hspace{-1pt} V^*(x,k) - \ell(x,u^*(0;x,k)) + \ell^*_{av}.
	\label{cost_decrease}
	\end{equation*}
	Since $\ell$ is bounded on $\Z$, $V^*(x_{\nu_k}(x,1),k+1) - V^*(x,k)$ is also bounded on $\mathcal{X}_{N(k)}$. Using the same analysis as in the proof of \cite[Theorem 1]{Angeli12} yields the claim.
\end{proof}

\section{Asymptotic Stability} \label{sect_stability}

In this section, we present constructive conditions under which the economic MPC with cyclic horizon achieves stability of the control invariant set $\bar{\X}$, in addition to convergence. To this end, we first state a sufficient condition for asymptotic stability of $\bar{\X}$.

\begin{corollary}
	Under the conditions of Theorem \ref{thm_convergence}, the control invariant set $\bar{\X}$ is asymptotically stable with a region of attraction $\mathcal{X}_{\hat{N}}$ if there exists a $\mathcal{K}_\infty$-function $\sigma$ such that
	\begin{equation}
	\bar{V}^*(x(0),0)\le\sigma(||x(0)||_{\bar{\X}}), \quad\forall x(0) \in \mathcal{X}_{\hat{N}}.
	\label{ub_rot_value_fcn}
	\end{equation}
	\label{cor_stability}
\end{corollary}
\begin{proof}
	Since convergence to $\bar{\X}$ in $\mathcal{X}_{\hat{N}}$ is shown in Theorem \ref{thm_convergence}, we only need to verify stability.  Let $\epsilon>0$ be arbitrary and $\delta\coloneqq\sigma^{-1}(\rho(\epsilon))$. Consider an $x(0)\in\mathcal{X}_{\hat{N}}$ with $||x(0)||_{\bar{\X}}<\delta$ such that by \eqref{ub_rot_value_fcn}, $\bar{V}^*(x(0),0)\le\sigma(\delta)$. From \eqref{cost_decrease_rot_2} and Corollary \ref{cor_rot_cost}, $\{\bar{V}^*(x_\nu(x(0),k),k)\}_{k\in\I_{\ge 0}}$ is a non-increasing sequence, i.e.,
	\begin{equation*}
	\bar{V}^*(x_\nu(x(0),k),k)\le\bar{V}^*(x(0),0), \quad \forall k\in\I_{\ge 0}.
	\end{equation*}
	We also immediately have the lower bound $\bar{V}^*(x,k)\ge L(x,u)\ge\rho(||x||_{\bar{\X}})$ from Assumption \ref{ass_lb_term_cost_rot} and Corollary \ref{cor_rot_cost}. Using these two properties yields
	\begin{align*}
	||x_\nu(x(0),k)||_{\bar{\X}}&\le \rho^{-1}(\bar{V}^*(x_\nu(x(0),k),k)) \\
	&\le \rho^{-1}(\bar{V}^*(x(0),0))\le\rho^{-1}(\sigma(\delta))=\epsilon
	\end{align*}
	for all $k\in\I_{\ge 0}$, which proves stability of $\bar{\X}$ in $\mathcal{X}_{\hat{N}}$.
\end{proof}

In \cite{Koegel13} and \cite{Lazar15}, a $\mathcal{K}_\infty$ upper bound on the value function is assumed for all $k$, which is known as uniform weak controllability (\cite{Rawlings09}). On the other hand, since the predictive controller is initialized with a horizon $N(0)=\hat{N}$ at initial time $0$, Corollary \ref{cor_stability} shows that an upper bound for $k=0$ as in \eqref{ub_rot_value_fcn} is sufficient for (non-uniform) stability. Still it is difficult, if not impossible to verify \eqref{ub_rot_value_fcn} directly, since the rotated value function is in general not known beforehand. To this end, we will derive in the following sufficient conditions that guarantee \eqref{ub_rot_value_fcn} for the MPC with cyclic horizon. 

\begin{lemma}
	If Assumption \ref{ass_lb_term_cost_rot} holds, then the rotated terminal cost is upper bounded by a $\mathcal{K}_\infty$-function $\alpha$, i.e.,
	\begin{equation*}
	\bar{V}_f(x) \le \alpha(||x||_{\bar{\X}}), \quad \forall x\in\X_f.
	\end{equation*}
	\label{lemma_lb_ub_term_cost_rot}
\end{lemma}
\begin{proof}
	\fontdimen2\font=3pt The upper bound follows directly from \cite[Lemma 12]{Amrit11}, Assumption \ref{ass_lb_term_cost_rot} and continuity of $\bar{V}_f$.
\end{proof}
\begin{remark}
	Lemma \ref{lemma_lb_ub_term_cost_rot} together with Assumption \ref{ass_lb_term_cost_rot} and \eqref{term_decrease_rot} implies that inside $\X_f$, $\bar{V}_f$ is a "finite-step" Lyapunov function as defined in \cite{Geiselhart14}.
\end{remark}
\begin{remark}
	To establish convergence, it would be sufficient that $\bar{\X}\supseteq\arg\min_{x}\bar{V}_f(x)$ (cf. Assumption \ref{ass_lb_term_cost_rot}). However, it is apparent that the upper bound on $\bar{V}_f$ cannot hold if this weaker condition is satisfied with a strict set inclusion.
\end{remark}

\begin{lemma}
	Suppose that Assumptions \ref{ass_dissi}-\ref{ass_lb_term_cost_rot} are fulfilled. Then, if $\hat{N} = JM, \; J\in\I_{\ge 1}$, it holds that
	\begin{equation*}
	\bar{V}^*(x,0)\le\alpha(||x||_{\bar{\X}}), \quad\forall x\in\X_f.
	\end{equation*}
	\label{lemma_ub_value_fcn_term_region}
\end{lemma}
\begin{proof}
	Consider the value functions
	\begin{equation*}
	\bar{V}^*_{jM}(x) \coloneqq \min_{u(\cdot)\in\U^{jM}(x)} \{\bar{V}(x,u(\cdot),jM)| x_u(x,jM)\in\X_f\}
	\end{equation*}
	for $j\in \I_{[0,J]}$, and their feasible sets $\mathcal{X}_{jM}$. From the dynamic programming recursion, we have for $j=1$
	\begin{equation}
	\begin{aligned}
	\bar{V}^*_M(x) = \min_{u(\cdot)\in\U^{M}(x)} &\{\sum_{i=0}^{M-1}L(x_u(x,i),u(i)) \label{prf_ub_value_fcn} \\
	+ &\bar{V}^*_0(x_u(x,M))| x_u(x,M)\in\mathcal{X}_0\}
	\end{aligned}
	\end{equation}
	with $\bar{V}^*_0(x)\hspace{-1pt}=\hspace{-1pt}\bar{V}_f(x)$ and $\mathcal{X}_0\hspace{-1pt}=\hspace{-1pt}\X_f$. By Assumption \ref{ass_term_decrease} and \eqref{term_decrease_rot}
	\begin{align*}
	\bar{V}^*_M(x) \hspace{-1pt}\stackrel{\eqref{prf_ub_value_fcn}}{=}\hspace{-1pt} \min \{\ldots\} \le& \hspace{-1pt} \sum_{i=0}^{M-1} \hspace{-2pt} L(x_\kappa(x,i),\kappa_i(x_\kappa(x,i))) \\
	+& \bar{V}_f(x_\kappa(x,M)) \le \bar{V}_f(x), \quad \forall x\in\X_f. \nonumber
	\end{align*}
	With ``$M$-step invariance'' of $\X_f$ from Assumption \ref{ass_term_inv}, we have $\X_f\subseteq\mathcal{X}_M\subseteq\ldots\subseteq\mathcal{X}_{JM}$. Then, analogous to the proof of \cite[Lemma 2.15]{Rawlings09}, we obtain
	\begin{equation*}
	\bar{V}^*_{(j+1)M}(x)\le\bar{V}^*_{jM}(x), \quad\forall x\in\mathcal{X}_{jM}, \; j\in\I_{[0,J-1]}.
	\end{equation*}
	Since $N(0)=\hat{N}=JM$, we finally obtain using Lemma \ref{lemma_lb_ub_term_cost_rot}
	\begin{equation*}
	\bar{V}^*(x,0) = \bar{V}^*_{JM}(x) \le \bar{V}_f(x) \le \alpha(||x||_{\bar{\X}}), \quad \forall x\in\X_f.
	\end{equation*}
\end{proof}
\begin{assumption}
	For $\X$, $\X_f$ and $\bar{\X}$ it holds that
	\begin{enumerate}[label=\alph*)]
		\item \label{ass_interior_a} $\text{int}(\X_f)\neq\emptyset$,
		\item \label{ass_interior_b} $\exists \hat{a}>0$ such that $\forall a\in[0,\hat{a}):$
		\begin{equation}
		(\bar{\X}\oplus \mathcal{B}_a)\cap \X_f = (\bar{\X}\oplus \mathcal{B}_a)\cap \X.
		\label{interior_cond}
		\end{equation}
	\end{enumerate}
	\label{ass_interior}
\end{assumption}
Assumption \ref{ass_interior}\ref{ass_interior_b} means that in an $a$-neighborhood around $\bar{\X}$, the boundaries of the terminal region and the constraint set are the same. It excludes $\bar{\X}=\X_f\subset\X$ and implies that if there are points in $\bar{\X}$ which lie on the boundary of $\X_f$, they must also lie on the boundary of $\X$. The notion of Assumption \ref{ass_interior} is more general than that of $\bar{\X}\subseteq\text{int}(\X_f)$, what is typically assumed to extend the upper bound on the value function to the entire feasible set, using a terminal region and cost (cf. \cite{Rawlings12}). A configuration for which $\bar{\X}\not\subseteq\text{int}(\X_f)$, but for which Assumption \ref{ass_interior} holds, is depicted in Figure \ref{interior_fig}.
\begin{figure}[h]
	\centering
	\includegraphics[scale=1.1]{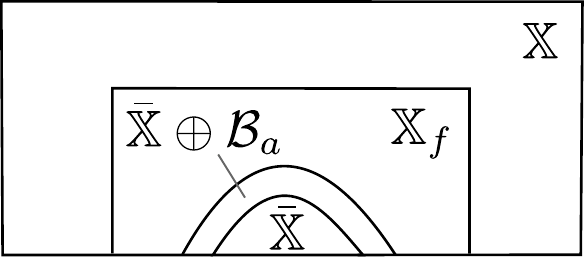}
	\caption{$\X$, $\X_f$ and $\bar{\X}$ fulfilling Assumption \ref{ass_interior}.}
	\label{interior_fig}
\end{figure}
\begin{lemma}
	Suppose the conditions of Lemma \ref{lemma_ub_value_fcn_term_region} and Assumption \ref{ass_interior} are fulfilled. Then, if the input constraint set $\U$ is compact, there exists a $\mathcal{K}_\infty$-function $\sigma$ such that
	\begin{equation*}
	\bar{V}^*(x,0)\le\sigma(||x||_{\bar{\X}}),\quad\forall x\in\mathcal{X}_{\hat{N}}.
	\end{equation*}
	\label{lemma_ub_value_fcn_feas_set}
\end{lemma}
\begin{proof}
	With compactness of $\U$, $\bar{V}^*(x,0)$ is locally upper bounded on $\mathcal{X}_{\hat{N}}$ by \cite[Proposition 1]{Rawlings12}. In the following, we construct $\sigma$.
	
	Note that Assumption \ref{ass_interior}\ref{ass_interior_b} implies that \eqref{interior_cond} also holds with $\mathcal{X}_{\hat{N}}$ instead of $\X$, since $\X_f\subseteq\mathcal{X}_{JM}=\mathcal{X}_{\hat{N}}\subseteq\X$. Given Assumption \ref{ass_interior}, there exists an $a>0$ such that $S_1 \coloneqq (\bar{\X}\oplus \mathcal{B}_a)\cap \mathcal{X}_{\hat{N}}$ is a subset of $\X_f$, i.e., $S_1\subseteq\X_f$. Then, by Assumption \ref{ass_interior}\ref{ass_interior_b}, $\mathcal{X}_{\hat{N}}\setminus S_1$ contains only points which fulfill $||x||_{\bar{\X}}>a>0$. This means that points in $\mathcal{X}_{\hat{N}}\setminus S_1$ cannot lie arbitrarily close to $\bar{\X}$. Instead, $\alpha(a)$ is an upper bound on $\bar{V}^*(x,0)$ for all $x\in (\bar{\X}\oplus \mathcal{B}_a)\cap \mathcal{X}_{\hat{N}}\subseteq\X_f$.
	
	The definition of $S_i$, $i\in\I_{\ge 2}$ and the remainder of the construction of $\sigma$ is equivalent to the proof of \cite[Proposition 11]{Rawlings12}.
\end{proof}

Next, we state the main result of this section, which is a direct combination of Corollary \ref{cor_stability} and Lemmas \ref{lemma_lb_ub_term_cost_rot}-\ref{lemma_ub_value_fcn_feas_set}.

\begin{theorem}
	Suppose that Assumptions \ref{ass_dissi}-\ref{ass_interior} hold, that $\hat{N}=JM, \; J\in\I_{\ge 1}$ and that the input constraint set $\U$ is compact. Then, the control invariant set $\bar{\X}$ is asymptotically stable with a region of attraction $\mathcal{X}_{\hat{N}}$.
	\label{thm_stability}
\end{theorem}

\section{Application: Control over Network} \label{sect_NCS}

As mentioned in the introduction, the development of the preceding analysis was motivated by a problem arising in the context of NCS. We briefly discuss how to apply the presented theory to this special case, while a more detailed exposition in the context of NCS is discussed in \cite{Wildhagen19b}. The considered setup is to control a discrete-time, linear time-invariant plant
\begin{equation}
x_p(k+1) = Ax_p(k) + Bu_p(k)
\label{lin_dyn}
\end{equation}
subject to the constraints $x_p(k)\in\X_p\subseteq\mathbb{R}^{n_p}$ and $u_p(k)\in\U_p\subseteq\mathbb{R}^{m_p}$ over a known, deterministic network. Transmissions over the network are required to fulfill a so-called token-bucket specification, which represents the communication capacities of the network. This specification was first introduced in the context of networked control in \cite{Linsenmayer18}; a more general characterization can be found, e.g., in \cite{Tanenbaum11}. The level of the bucket evolves according to the saturating dynamics
\begin{equation}
\beta(k+1) = \min\{\beta(k)+g-\gamma(k)c,b\},
\label{token_bucket}
\end{equation}
where $\beta(k)$ is the current bucket level and $\gamma(k)\in\{0,1\}$ is the decision on whether to transmit over the network or not. The parameters are the token generation rate $g\in\I_{\ge 1}$, the cost per transmission $c\in\I_{\ge g}$ and the bucket size $b\in\I_{\ge c}$. A transmission sequence $\gamma(\cdot)$ that fulfills the token-bucket specification may never drain it, i.e., $\beta(k)\ge 0$ for all $k\in\I_{\ge 0}$ under \eqref{token_bucket}. It is typically assumed that $c>g$, i.e., it is not possible to transmit at every time instance. The plant \eqref{lin_dyn} receives a new control input $u_c$ only if one is sent over the network, otherwise the last applied input is held. Associated with the plant is the quadratic cost on the state and applied input $x_p^\top Qx_p + u_p^\top Ru_p$ ($Q,R>0$), while the cost is independent of the bucket level.

The applied input from the last time step is saved in
\begin{equation}
u_s(k+1) = \gamma(k)u_c(k) + (1-\gamma(k))u_s(k) \eqqcolon u_p(k).
\label{ZOH}
\end{equation}
Denoting the overall state $x\coloneqq[x_p, u_s, \beta]$ and the control $u\coloneqq[u_c, \gamma]$, the economic cost of the overall system is
\begin{equation}
\ell(x,u) = x_p^\top Qx_p + u_c^\top \gamma Ru_c + u_s^\top(1-\gamma)Ru_s, \label{eco_cost_NCS}
\end{equation}
the state constraint set is $\X=\X_p\times\U_p\times\I_{[0,b]}$, the input constraint set is $\U=\U_p\times\{0,1\}$, and $f$ is composed of the right hand side of \eqref{lin_dyn}, \eqref{token_bucket} and \eqref{ZOH}. Due to its independence of the bucket level $\beta$, and the decision variable $\gamma$ appearing as a factor in its terms, the economic cost \eqref{eco_cost_NCS} is not positive definite. The lowest possible asymptotic average cost $\ell^*_{av}=0$ is attained in the control invariant set $\bar{\X}=\{0\}\times\{0\}\times\I_{[0,b]}$. Using the storage function $\lambda(x)=u_s^\top Su_s,\;R>S>0$ yields a rotated stage cost $L(x,u) \ge x_p^\top Q x_p + u_s^\top S u_s = \rho(||x||_{\bar{\X}})$, hence the system is dissipative with respect to $\bar{\X}$ and supply rate $\ell(x,u)-\ell^*_{av}$.

Due to the token-bucket network, it can only be guaranteed that a new control input can be sent to the plant every $M=\lceil \frac{c}{g}\rceil$ time instances. Hence, a compact terminal region $\X_f$, other that $\bar{\X}$, cannot be made control invariant for general unstable plants \eqref{lin_dyn}. It can only be guaranteed that the state returns to $\X_f$ after $M$ steps, which ensures recursive feasibility for an MPC with cyclic horizon and $\hat{N}\ge M$. In particular, assuming that the pair $(A^M,\sum_{i=0}^{M-1}A^i B)$ is controllable, there exist $V_f(x)=x_p^\top Px_p, \; P>0$ and $K$ such that if the terminal controller is chosen as
\begin{equation*}
\kappa_0(x) = \begin{cases} [Kx_p, 1] & \beta\ge c-g \\
[0,0] & \text{otherwise}				\end{cases}
\end{equation*}
and $\kappa_i(x)=[0, 0], \; i\in\I_{[1,M-1]}$, then there exists an $a\ge0$ such that with the terminal region $\X_f = \{0\}\times\{0\}\times\I_{[0,c-g)}\cup\{x_p|x_p^\top P x_p \le a\}\times\U_p\times\I_{[c-g,b]}$, Assumptions \ref{ass_term_inv} and \ref{ass_term_decrease} are fulfilled. The rotated terminal cost $\bar{V}_f(x)=x_p^\top Px_p + u_s^\top Su_s$ attains its minimum exactly on $\bar{\X}$. Hence, all conditions of Theorem \ref{thm_convergence} are fulfilled, and provided initial feasibility, the state of the closed loop system converges to $\bar{\X}$ as $k\rightarrow\infty$.

Note that Assumption \ref{ass_interior} does not hold in this example. However, consider the set $\mathcal{X}^*_{\hat{N}}\coloneqq \{x\in\mathcal{X}_{\hat{N}}| [x_p,u_s] = [0,0] \text{ or } \beta \in \I_{[c-g,b]}\}$, for which Assumption \ref{ass_interior} holds (if $\X$ is replaced by $\mathcal{X}^*_{\hat{N}}$). Then, if $\hat{N}=JM$ and $\U_p$ is compact, an upper bound of $\bar{V}^*(x,0)$ in $\mathcal{X}^*_{\hat{N}}$ is ensured by Lemmas \ref{lemma_lb_ub_term_cost_rot}, \ref{lemma_ub_value_fcn_term_region} and \ref{lemma_ub_value_fcn_feas_set}. Therefore, asymptotic stability is guaranteed if the initial state lies in $\mathcal{X}^*_{\hat{N}}$.

For numerical simulation, we consider the linearized and discretized batch reactor taken from \cite{Heemels13}, which is a well-known benchmark example in NCS. The matrices $A$ and $B$ can be found therein. We consider the box constraints $\X_p=[-1.2,1.2]^{4}$ and $\U_p=[-2,2]^{2}$, and the initial plant state $x_p(0)=[1,0,1,0]$. The bucket with $g=1$, $c=3$ (such that $M=3$) and $b=10$ is initialized at $\beta(0)=2$, and we set $u_s(0)=0$. The cost matrices are chosen as $Q=10I$ and $R=I$. Figure \ref{plant_state} shows that the set $\bar{\X}$ is asymptotically stable for the overall system.
\begin{figure}[h]
	\centering
%
%

\pgfplotsset{every tick label/.append style={font=\small}}

\begin{tikzpicture}

\begin{axis}[%
width=\columnwidth,
height=5.5cm,
at={(1.609in,0.442in)},
xmin=0,
xmax=14,
xlabel={Time $k$},
ymin=-0.7,
ymax=1.3,
ylabel={Plant state $x_s$},
axis background/.style={fill=white},
legend style={legend cell align=left, align=left, draw=white!15!black, font=\small}
]
\addplot[const plot, color=istgreen] table[row sep=crcr] {%
0	0.817766041686343\\
1	1.19999993008871\\
2	1.11685911072027\\
3	0.714513973501339\\
4	0.511814223610472\\
5	0.245779776503706\\
6	0.114311908525181\\
7	0.0498215858135726\\
8	0.0260503515981224\\
9	0.00953483350681749\\
10	0.00018514908152721\\
11	-0.0015470790730304\\
12	-0.000725743100427778\\
13	0.000122880956697032\\
14	0.000245017222447768\\
};
\addlegendentry{$x_{p1}$}

\addplot[const plot, color=istorange] table[row sep=crcr] {%
0	0.0239957261687133\\
1	-0.0104388638216448\\
2	0.0218201976814627\\
3	0.0374891963741133\\
4	0.0313796916487677\\
5	0.0770302468436573\\
6	0.115311526902875\\
7	0.041004518283759\\
8	0.00514644652846105\\
9	-0.0131322844747161\\
10	-0.00497789619174662\\
11	-0.000678450043002867\\
12	0.00161302152423128\\
13	0.000526597919374484\\
14	-4.88237211740977e-05\\
};
\addlegendentry{$x_{p2}$}

\addplot[const plot, color=istblue] table[row sep=crcr] {%
0	0.924820914310183\\
1	0.717544661560291\\
2	0.116293145185841\\
3	-0.112485104035186\\
4	-0.580123540030378\\
5	-0.270591356370287\\
6	-0.134123348835291\\
7	-0.0263410149893854\\
8	-0.00489113821219196\\
9	-0.0136028479873186\\
10	-0.00448020945716724\\
11	-0.00113490152493353\\
12	0.00153100862085881\\
13	0.000692326965584743\\
14	0.00024823384756101\\
};
\addlegendentry{$x_{p3}$}

\addplot[const plot, color=istred] table[row sep=crcr] {%
0	0.0904230982682515\\
1	0.155270544751583\\
2	0.108292880092153\\
3	-0.0519063879650103\\
4	-0.0703653055040311\\
5	-0.0655519348670312\\
6	-0.00570828030095057\\
7	0.0206040275367803\\
8	0.0206836218775896\\
9	0.00713901250813132\\
10	-0.000294929850129803\\
11	-0.00157571121882485\\
12	-0.000300885671168094\\
13	0.000415827227396515\\
14	0.000409321781016464\\
};
\addlegendentry{$x_{p4}$}

\end{axis}
\end{tikzpicture}%
\hspace{5pt}
\begin{tikzpicture}

\begin{axis}[%
width=\columnwidth,
height=5.5cm,
at={(1.609in,0.442in)},
xmin=0,
xmax=14,
xlabel={Time $k$},
ymin=-0.5,
ymax=2.5,
ylabel={Control input $u_c$},
axis background/.style={fill=white},
legend style={legend cell align=left, align=left, draw=white!15!black, font=\small}
]
\addplot[const plot, color=istgreen] table[row sep=crcr] {%
	0	-0.000104842289406691\\
	1	0.114314419333174\\
	2	3.43691908801447e-05\\
	3	-6.70599091835802e-21\\
	4	0.140528622737944\\
	5	-6.70599091835802e-21\\
	6	-0.025134092531049\\
	7	-6.34741588783969e-18\\
	8	-6.34741588783969e-18\\
	9	0.00312827413835817\\
	10	-3.20135608894077e-18\\
	11	-3.20135608894077e-18\\
	12	-0.000548820933140323\\
	13	6.38869754737727e-18\\
	14	6.38869754737727e-18\\
};
\addlegendentry{$u_{c1}$}

\addplot[const plot, color=istorange] table[row sep=crcr] {%
	0	3.43580937611847e-05\\
	1	2.00145986383691\\
	2	3.43695724256873e-05\\
	3	-6.70599091835802e-21\\
	4	0.295146391652981\\
	5	-6.70599091835802e-21\\
	6	0.0526787353224282\\
	7	-6.34741588783969e-18\\
	8	-6.34741588783969e-18\\
	9	-0.00656740934645204\\
	10	-3.20135608894077e-18\\
	11	-3.20135608894077e-18\\
	12	0.000962738876684808\\
	13	6.38869754737727e-18\\
	14	6.38869754737727e-18\\
};
\addlegendentry{$u_{c2}$}

\end{axis}
\end{tikzpicture}%
	\caption{The evolution of controlled plant states and the control input over time with $\hat{N}=3$. If the control input is zero in the figure, no new input is sent over the network and the previously applied input is held.}
	\label{plant_state}
\end{figure}
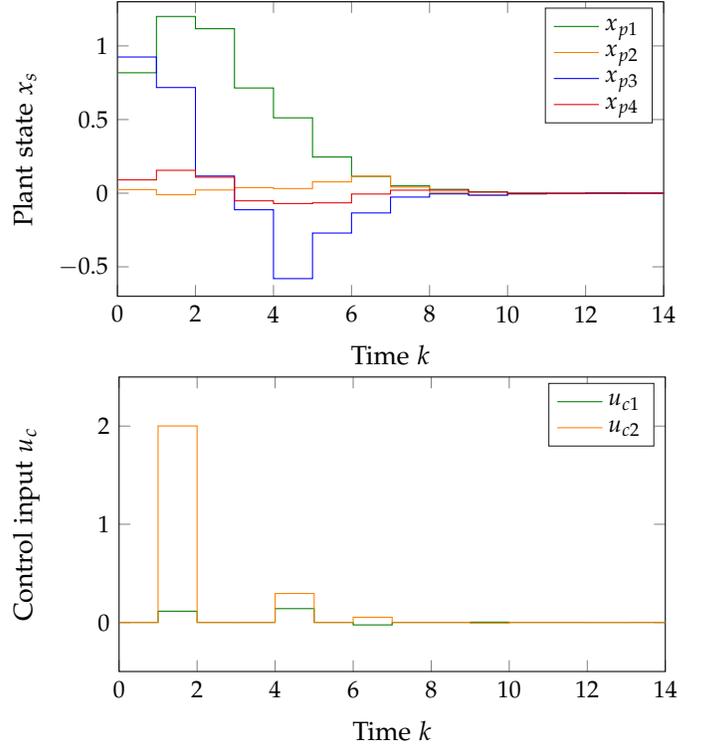
\section*{Acknowlegdements}
The first author would like to thank Johannes K\"{o}hler for many helpful comments and fruitful discussions.

\small
\bibliographystyle{plain}
\bibliography{bib_NOLCOS19}             
                                                
\end{document}